%% file: praca.tex
\documentclass[11pt, a4paper, final]{article}
\usepackage{fullpage} 
\usepackage{wasysym}

\usepackage[utf8]{inputenc}
\usepackage[numbers,sort]{natbib}
\usepackage{float}
\usepackage{caption}   %
\usepackage[english]{babel}
\usepackage[ruled,vlined]{algorithm2e}
\usepackage{amsfonts}

\usepackage{amsmath}
\usepackage{amsthm}
\usepackage{thm-restate}
\usepackage{enumitem}
\usepackage{tikz-cd}
\usepackage{tikz}
\usetikzlibrary{arrows,arrows.meta}
\usetikzlibrary{shapes.geometric}
\tikzcdset{arrow style=tikz, diagrams={>=latex}}
\usetikzlibrary{positioning, calc, shapes.misc}
\usepackage{etoolbox}
\usepackage[notref, notcite]{showkeys}
\usepackage{xcolor}
\usepackage[colorlinks=true]{hyperref}
\hypersetup{
	colorlinks,
	linkcolor={red!50!black},
	citecolor={blue!50!black},
	urlcolor={blue!80!black}
}
\usepackage[capitalize,nameinlink,noabbrev]{cleveref}

\newcommand\blfootnote[1]{%
  \begingroup
  \renewcommand\thefootnote{}%
  \NoHyper\footnote{#1}\endNoHyper
  \addtocounter{footnote}{-1}%
  \endgroup
}

\DeclareMathOperator*{\argmax}{arg\,max}

\newcommand{\pmms}{\mu}

\newtheorem{other}{\relax}[section] 
\newtheorem{lemma}[other]{Lemma}

\newtheorem{proposition}[other]{Proposition}

\newtheorem{corollary}[other]{Corollary}
\newtheorem{remark}[other]{Remark}
\newtheorem{example}[other]{Example}
\newtheorem{definition}[other]{Definition}

\AtEndEnvironment{proof}{\setcounter{clm}{0}}

\newcommand{\problemspace}{0.1cm}

\title{Probing EFX via PMMS:\\(Non-)Existence Results in Discrete Fair Division}
\author{Jarosław Byrka$^\ast$ \and Franciszek Malinka$^\dagger$ \and Tomasz Ponitka$^\ddagger$}

\date{July 30, 2025}

\begin{document}

\maketitle

\blfootnote{This project has been partially funded by NCN grant number 2020/39/B/ST6/01641, by the European Research Council (ERC) under the European Union's Horizon 2020 research and innovation program (grant agreement No. 866132), by an Amazon Research Award, by the Israel Science Foundation Breakthrough Program (grant No. 2600/24), by a grant from the TAU Center for AI and Data Science (TAD), and by the NSF-BSF (grant number 2020788).}
\blfootnote{$^\ast$University of Wrocław, Poland. Email: \texttt{jby@cs.uni.wroc.pl}}
\blfootnote{$^\dagger$University of Wrocław, Poland. Email: \texttt{franciszek.malinka@gmail.com}}
\blfootnote{$^\ddagger$Tel Aviv University, Israel. Email: \texttt{tomaszp@mail.tau.ac.il}}

\begin{abstract}

We study the fair division of indivisible items and provide new insights into the EFX problem, which is widely regarded as the central open question in fair division, and the PMMS problem, a strictly stronger variant of EFX. Our first result constructs a three-agent instance with two monotone valuations and one additive valuation in which no PMMS allocation exists. Since EFX allocations are known to exist under these assumptions, this establishes a formal separation between EFX and PMMS. 

We prove existence of fair allocations for three important special cases. We show that EFX allocations exist for personalized bivalued valuations, where for each agent $i$ there exist values $a_i > b_i$ such that agent $i$ assigns value $v_i(\{g\}) \in \{a_i, b_i\}$ to each good $g$. We establish an analogous existence result for PMMS allocations when $a_i$ is divisible by $b_i$. We also prove that PMMS allocations exist for binary-valued MMS-feasible valuations, where each bundle $S$ has value $v_i(S) \in \{0, 1\}$. Notably, this result holds even without assuming monotonicity of valuations and thus applies to the fair division of chores and mixed manna. Finally, we study a class of valuations called pair-demand valuations, which extend the well-studied unit-demand valuations to the case where each agent derives value from at most two items, and we show that PMMS allocations exist in this setting. Our proofs are constructive, and we provide polynomial-time algorithms for all three existence results.

\end{abstract}

\newpage

\section{Introduction}
\label{section:introduction}

The fair division problem involves allocating a set of resources among a group of agents to meet certain fairness criteria. When there are only two agents, the standard fair division protocol is the {\em cut-and-choose} method---where one agent proposes a division of the resources into two parts, and the other selects the preferred part---a practice that dates back to the story of Abraham and Lot in the Book of Genesis. The fair division literature, beginning with \citet*{steinhaus}, explores the fairness guarantees that can be achieved for more than two agents.

The initial focus of this literature has been on the cake-cutting scenario, where the resources to be allocated are assumed to be infinitely divisible. A key property of cut-and-choose in cake-cutting scenarios with two agents, formalized by \citet*{foley1967resource}, is {\em envy-freeness}: no agent prefers another agent’s allocation to their own. A central achievement of the cake-cutting literature is that this guarantee can be extended to any number of agents, under mild assumptions \cite{stromquist1980how}.

This work focuses on the discrete fair division problem, where the goods are {\em indivisible}. 
Specifically, we consider a setting with $n$ agents and $m$ items, where each agent $i \in N$ has a valuation function $v_i : 2^{M} \to \mathbb{R}_{\geq 0}$ that assigns a non-negative value to each subset of the items. The goal is to compute a fair partition $(X_1, \ldots, X_n)$ such that $X_i \cap X_j = \emptyset$ for all $i \neq j$ and $X_1 \cup \cdots \cup X_n = M$.
It is easy to see that even for two agents, an envy-free allocation may not exist—for instance, when there is only one item.

\vspace{0.1cm}
\noindent \textbf{EFX.}
A key property of cut-and-choose in discrete settings with two agents, identified by \citet*{unreasonable-fariness}, is {\em envy-freeness up to any good} (EFX): no agent prefers another’s allocation after the removal of any single item from that bundle, i.e.,  for all $i,j \in N$, it holds that $v_i(X_i) \geq v_i(X_j \setminus \{g\})$ for all $g \in X_j$. 
Extending the existence of EFX allocations to an arbitrary number of agents---analogous to the existence of envy-free allocations in cake cutting---remains a major open problem, even for additive valuations; notably, no counterexamples are known, even with arbitrary monotone valuations.

Nevertheless, significant progress has been made on the EFX problem for various special cases. \citet*{Amanatidis_2021} addressed the case of {\em bivalued valuations},
where each agent assigns one of two values to each item, i.e., there are some $a,b \geq 0$ such that $v_i(\{g\}) \in \{a, b\}$ for all $i \in N$ and $g \in M$, and $v_i(S) = \sum_{g \in S} v_i(\{g\})$ for all $S \subseteq M$. They provided two distinct proofs of the existence of EFX allocations for bivalued valuations: (1) an existential proof based on maximizing Nash welfare (i.e., the product of agents' valuations), and (2) a constructive proof via the Match-and-Freeze algorithm. Both proofs crucially rely on the parameters $a$ and $b$ being fixed across all agents' valuations.

Furthermore, \citet*{efx-3-agents,DBLP:journals/ior/AkramiACGMM25} established the existence of EFX allocations for {\em three agents} where two agents have arbitrary monotone valuations and the third agent's valuation satisfies a natural condition known as MMS-feasibility, which covers important classes of valuations, including additive; see \Cref{def:mmsfeasibility}. See also \Cref{sec:relatedwork} for additional discussion on the EFX problem.

\vspace{0.1cm}
\noindent \textbf{PMMS.}
While EFX and its relaxations have received significant attention in the fair division literature, \citet*{unreasonable-fariness} also formalized another compelling fairness criterion known as the {\em pairwise maximin share} (PMMS), a natural generalization of the guarantees of cut-and-choose.
Specifically, for any pair of agents $i$ and $j$, it requires that agent $i$ receives at least as much value as they would in the maximin partition of their combined items, i.e., the partition $(A, B)$ of $X_i \cup X_j$ that maximizes $\min\{v_i(A), v_i(B)\}$. 
This corresponds to the value one could guarantee as the cutter in the cut-and-choose protocol.

While PMMS allocations may not exist for arbitrary monotone valuations (see \Cref{prop:pmmsnonexistance}), it remains an open question whether they exist under MMS-feasible valuations, including additive ones.
Notably, under mild assumptions, any PMMS allocation is also EFX, so proving the existence of PMMS allocations would simultaneously resolve the EFX problem; see \Cref{sec:implicationsofpmms} for more details. However, PMMS is significantly less understood, and little is known about its existence in special cases.

An important advantage of PMMS over EFX is that its definition naturally extends to the fair division of {\em chores}, where agents have negative valuations that decrease as the number of items increases. In contrast, the definition of EFX must be adapted for chores by requiring that $v_i(X_i \setminus \{g\}) \geq v_i(X_j)$ for all $g \in X_i$. Moreover, \citet*{ChristoforidisS24} recently showed that EFX for chores does not exist when there are three agents with arbitrary monotonically decreasing valuations;
notably, these valuations violate MMS-feasbility.

The goal of this work is to advance our understanding of the existence of both EFX and PMMS allocations in key special cases of discrete fair division.

\subsection{Our Contribution}

In this work, we provide new insights into the existence of EFX and PMMS allocations, including one non-existence result and three existence results.

We begin by presenting our negative result for PMMS in the case of three agents.

\begin{restatable}[Separation between EFX and PMMS for Three Agents]{theorem}{thmseparation}
    There exists an instance with three agents, one with an additive valuation and two with arbitrary monotone valuations, for which no PMMS allocation exists.
    \label{thm:separation}
\end{restatable}

This result establishes a formal separation between the existence of EFX and PMMS allocations for the case of three agents. \citet*{DBLP:journals/ior/AkramiACGMM25} have shown that EFX is always guaranteed under the assumptions of the theorem. In contrast, we show that PMMS may fail to exist. {This indicates that proving the existence of PMMS for the case of three agents might require fundamentally different techniques.}

Our first positive result concerns the existence of EFX in the important case of \emph{personalized} bivalued valuations \cite[e.g.,][]{DBLP:conf/atal/0001LRS23,DBLP:conf/wine/BuLLST23,DBLP:journals/corr/abs-2505-22174}, where each agent $i \in N$ has values $a_i, b_i \geq 0$ such that $v_i(\{g\}) \in \{a_i, b_i\}$ for all $g \in M$.
We also consider the existence of PMMS allocations under \emph{factored} personalized bivalued valuations \cite[e.g.,][]{ebadian2022fairlyallocateeasydifficult,DBLP:conf/aaai/AkramiC0MSSVVW22}, where each $a_i$ is an integer multiple of $b_i$.
We establish the following theorem for this setting.

\begin{restatable}[Personalized Bivalued Valuations]{theorem}{thmbivalued}
    For any instance with personalized bivalued valuations, there exists an allocation that satisfies EFX. Moreover, if the valuations are also factored, then this allocation is guaranteed to satisfy PMMS.
    \label{thm:bivalued}
\end{restatable}

This generalizes the result of \citet*{Amanatidis_2021} in two ways. First, we consider the more general setting of personalized bivalued valuations. Second, we also establish the existence of a PMMS allocation under an additional assumption, rather than just an EFX allocation.

Our result is obtained by adapting the Match-and-Freeze algorithm to the setting of personalized bivalued valuations (see \Cref{alg:personalized-bivalued-pmms}). Notably, we also show that the argument based on maximizing Nash welfare does not extend to this setting (see \Cref{prop:mnw-fail-personalized}).

We note that the existence of EFX allocations for personalized bivalued valuations has also been independently discovered by \citet*{jin2025paretooptimalfairallocationspersonalized}.

Our second positive result addresses the natural class of {\em binary-valued valuations}, where each bundle of items is either desirable or non-desirable to an agent; i.e., 
$v_i(S) \in \{0,1\}$ for all $i \in N$ and $S \subseteq M$.
We prove the following:

\begin{restatable}[Binary-Valued Valuations]{theorem}{thmbinary}
    For any instance with binary-valued MMS-feasible valuations, there exists a PMMS allocation.
    \label{thm:binary}
\end{restatable}

To prove our result, we introduce a novel Cut-and-Choose-Graph procedure (see \Cref{alg:binary-pmms}), inspired by the Envy-Graph procedure of \citet*{lipton}.

A key feature of our proof is that it does not rely on the monotonicity of valuations. Therefore, it applies not only to instances involving goods, but also to those involving chores or mixed manna, i.e., a combination of goods and chores.

Finally, we study the class of pair-demand valuations, where each agent desires to receive at most two items. This class naturally generalizes unit-demand valuations and is a special case of the broader and well-studied class of $k$-demand valuations \cite[e.g.,][]{DBLP:conf/atal/DeligkasMS21,DBLP:conf/icml/ZhangC20}, with $k = 2$.
We establish the following theorem for this setting.

\begin{restatable}[Pair-Demand Valuations]{theorem}{thmpmmspairdemand}
    For any instance with  pair-demand valuations, there exists a PMMS allocation.
    \label{thm:pmms-pair-demand}
\end{restatable}

All of our positive results are not only existential but also constructive: every algorithm presented in this work runs in polynomial time (see \Cref{cor:compute_bivalued,cor:compute_binary,cor:compute_pair_demand}).

\subsection{Related Work}
\label{sec:relatedwork}

In this section, we review the related literature. For a broader overview of the field, see the survey by \citet*{AMANATIDIS2023103965}.

\paragraph{PMMS.} The PMMS notion was introduced by \citet*{unreasonable-fariness}. Since the existence of an exact PMMS allocation remains an open problem, many works have focused on multiplicative approximations of PMMS. An allocation is said to be $\alpha$-PMMS if $v_i(X_i) \geq \alpha \cdot \mu_i(X_i \cup X_j)$ for every pair of agents $i$ and $j$.

The best known approximation of PMMS guaranteed to exist for additive valuations is $0.781$-PMMS, due to the construction by \citet*{kurokawa-phd}. \citet*{unreasonable-fariness} also show that, for additive valuations, an allocation maximizing Nash welfare guarantees a $0.618$-PMMS. Furthermore, \citet*{Amanatidis_2020} show that, under additive valuations, there exists an allocation that is simultaneously $0.618$-EFX and $0.717$-PMMS. \citet*{DBLP:conf/aaai/FeldmanMP24} further demonstrate the existence of an allocation that is $0.618$-EFX, $0.618$-PMMS, and achieves a $0.618$-fraction of the maximum Nash welfare. Additionally, \citet*{DBLP:conf/ijcai/AmanatidisBM18} show that any EFX allocation is $2/3$-PMMS, and any EF1 allocation is $1/2$-PMMS.

PMMS has also been studied in settings with restricted valuation classes. \citet*{BarmanV21} show that a partial PMMS allocation with optimal utilitarian social welfare exists for matroid rank valuations. \citet*{DBLP:conf/atal/BarmanKP24} prove the existence of a $5/6$-PMMS allocation for restricted additive valuations. Moreover, a $4/5$-PMMS allocation exists when agents have additive valuations with identical ordering over goods \cite{DAI2024114388}.

\citet*{sun2023fairness} study the implications between PMMS and other fairness notions in the fair division of chores.

Finally, PMMS has been explored in three different graph-based settings. First, \citet*{DBLP:conf/atal/Hummel024} study PMMS when items form a graph and allocations are restricted to connected subgraphs. Second, \citet*{DAI2024114388} consider a setting with a given social network over agents, where comparisons between bundles are restricted to neighbors in the network; they show that PMMS exists for a path of size three. Third, \citet*{DBLP:journals/corr/abs-2506-20317} analyze a model where agents are nodes and items are edges in a graph, with each item being relevant only to the two agents it connects. They show that PMMS always exists for any multigraph.

\paragraph{EFX.} 
The notion of EFX was introduced by \citet*{unreasonable-fariness}. As with PMMS, significant attention has been devoted to studying multiplicative approximations of EFX. An allocation is said to be $\alpha$-EFX if $v_i(X_i) \geq \alpha \cdot v_i(X_j \setminus \{g\})$ for every pair of agents $i$ and $j$, and every $g \in X_j$.

The state-of-the-art approximation guarantees for EFX are $1/2$-EFX for subadditive valuations \cite{plaut2017envyfreenessgeneralvaluations}, and $0.618$-EFX for additive valuations \cite{Amanatidis_2020}. Further results on EFX approximations appear in \cite{DBLP:conf/sigecom/AmanatidisFS24, DBLP:conf/atal/MarkakisS23, DBLP:conf/atal/BarmanKP24}. The compatibility of (approximate) EFX with efficiency has been investigated in the context of Nash welfare maximization \cite{DBLP:conf/ec/CaragiannisGH19, DBLP:conf/stoc/GargHLVV23, DBLP:journals/corr/abs-2407-12461, DBLP:conf/aaai/FeldmanMP24}, and utilitarian welfare maximization \cite{DBLP:journals/mst/BeiLMS21, DBLP:conf/wine/BarmanB020, DBLP:conf/sigecom/Bu0LST25, DBLP:conf/atal/LiLLTT24}.

EFX has been shown to exist in several important special cases. Notably, EFX allocations exist for three agents \cite{efx-3-agents, DBLP:journals/ior/AkramiACGMM25}, and more generally, for three types of agents \cite{DBLP:conf/sigecom/HVGN025, DBLP:journals/dam/Mahara23}. Moreover, EFX exists under (non-personalized) bivalued valuations, as shown by \citet*{Amanatidis_2021}. \citet*{DBLP:conf/faw/BuSY23} show that EFX always exists for instances with binary-marginal valuations, where $v_i(S \cup \{g\}) - v_i(S) \in \{0,1\}$ for every bundle $S$ and item $g$. The binary-marginal case has been further explored in \cite{DBLP:conf/aaai/BabaioffEF21, DBLP:conf/wine/0002PP020, DBLP:journals/teco/BenabbouCIZ21}.

While constructing a complete EFX allocation is a challenging problem, a partial EFX allocation is trivially achievable; for instance, an empty allocation satisfies EFX. This is often interpreted as setting aside some items and giving them to charity \cite{DBLP:conf/ec/CaragiannisGH19}. Beginning with the work of \citet*{DBLP:journals/siamcomp/ChaudhuryKMS21}, several studies have focused on partial EFX allocations that minimize the number of unallocated items \cite{DBLP:conf/aaai/BergerCFF22, DBLP:conf/sigecom/ChaudhuryGMMM21, DBLP:conf/mfcs/BerendsohnBK22, DBLP:conf/ijcai/JahanSJS23, DBLP:journals/ior/AkramiACGMM25}.

EFX has also been examined in the context of ``best-of-both-worlds'' guarantees for randomized lotteries over allocations. The goal here is to construct (approximately) envy-free distributions over (approximately) EFX allocations \cite{DBLP:conf/sigecom/FeldmanMNP24, DBLP:journals/corr/abs-2410-06877, DBLP:journals/corr/abs-2410-08986}.

Finally, EFX has been studied in chore allocation problems. \citet*{DBLP:conf/stoc/GargMQ25} demonstrate that a constant-factor approximation of EFX exists for additive cost functions, whereas \citet*{ChristoforidisS24} provide an example with monotone cost functions for which no EFX allocation exists. Additional results on EFX for chores can be found in \cite{DBLP:journals/ai/ZhouW24, DBLP:journals/tcs/KobayashiMS25, DBLP:conf/atal/BarmanNV23, DBLP:conf/atal/0001LRS23, DBLP:conf/www/0037L022, DBLP:conf/ijcai/GargMQ23, DBLP:journals/corr/abs-2410-18655}.

\paragraph{Other Fairness Notions.}The fair division literature has also examined other fairness guarantees. Similar to EFX and PMMS, there is extensive research on approximate MMS \cite{DBLP:journals/jacm/KurokawaPW18,DBLP:journals/corr/abs-2502-09377,DBLP:journals/teco/BarmanK20,DBLP:conf/sigecom/GhodsiHSSY18,DBLP:conf/soda/GargMT19,DBLP:journals/talg/AmanatidisMNS17,DBLP:journals/ai/GargT21,DBLP:conf/ijcai/AkramiGST23}. The current best-known approximation is $3/4 + \Omega(1)$, due to \citet*{DBLP:conf/soda/AkramiG24}. In contrast to EFX and PMMS, it is known that MMS allocations may not always exist, and the best-known upper bound on the approximation ratio is $39/40$, established by \citet*{DBLP:conf/wine/FeigeST21}.

\citet*{DBLP:conf/aaai/BarmanBMN18} introduced the notion of groupwise maximin share (GMMS), which strengthens both PMMS and MMS. Since GMMS is stronger than MMS, it also may not always exist. Approximation results for GMMS have been studied in \cite{Amanatidis_2020,DBLP:journals/siamcomp/ChaudhuryKMS21,DBLP:conf/wine/BarmanV21}.

Furthermore, \citet*{DBLP:conf/ijcai/CaragiannisGRSV23} introduced a relaxation of EFX called Epistemic-EFX, whose existence under all monotone valuations was proven by \citet*{DBLP:conf/aaai/AkramiR25}. Finally, the study of fair division with subsidies, where fairness is achieved via monetary transfers among agents, was initiated by \citet*{DBLP:conf/sagt/HalpernS19} and has since been further explored by \cite{DBLP:conf/sigecom/BrustleDNSV20,DBLP:conf/aaai/FeldmanGNP25,DBLP:conf/ijcai/BarmanKNS22,DBLP:conf/wine/CaragiannisI21}.

\section{Preliminaries}
\label{section:preliminaries}

In this work we consider instances of the fair division problem where $N$ is a set of $n$ \emph{agents}, $N = \{1, \ldots, n\}$, and $M$ is a set of $m$ indivisible \emph{items}, $M = \{1, \ldots, m\}$. Each agent $i \in N$ has a valuation function $v_i : 2^M \to \mathbb{R}$. 

Unless stated otherwise, we assume that each agent $i$'s valuation is normalized, so that $v_i(\emptyset) = 0$, and monotone, so that $v_i(S) \le v_i(T)$ for any $S \subseteq T \subseteq M$.  Notably, we do not make these assumptions in \Cref{sec:binary}.

An \emph{allocation} $X = \langle X_1, \ldots, X_n \rangle$ is a partition of the items among the agents, where each item is assigned to exactly one agent. Formally, we require that  $X_1 \cup \cdots \cup X_n = M$ and $X_i \cap X_j = \emptyset$ for all $i \neq j$.

\subsection{Fairness Notions}

We now formally define the EFX condition.

\begin{definition}[EFX Allocations]
	\label{def:efx}
    We say that an allocation $X$ is \emph{envy-free up to any good} (EFX) if, for every pair of agents $i$ and $j$, and for every item $g \in X_j$, we have $v_i(X_i) \ge v_i(X_j \setminus \{g\})$.
    Otherwise, if there exists some item $g \in X_j$ such that $v_i(X_i) < v_i(X_j \setminus \{g\})$, we say that agent $i$ \emph{EFX-envies} agent $j$.
\end{definition}

To introduce the notion of PMMS, we first define the \emph{fair share function} $\mu_i$.

\begin{definition}[Fair Share]
\label{def:mu}
Let $k \in \{2, \ldots, n\}$. The \emph{fair share} of agent $i$ for a bundle $S \subseteq M$ divided among $k$ agents is defined as
\begin{align*}
\mu_i(S, k) = \max_{\substack{X_1, \ldots, X_k \ \text{partition of } S}} \min \{ v_i(X_1), \ldots, v_i(X_k) \}.
\end{align*}
We refer to any partition $X_1, \ldots, X_k$ that achieves the maximum above as a \emph{maximin partition}.
If only a single agent is being considered and it is clear from context, we omit the subscript and write $\mu$ instead of $\mu_i$. Similarly, when $k = 2$, we omit the parameter and write $\mu_i(S)$ to denote $\mu_i(S, 2)$.
\end{definition}

Using the notion of fair shares, we define the PMMS condition.

\begin{definition}[PMMS Allocations]
\label{def:pmms}
An allocation $X$ is \emph{pairwise maximin share fair} (PMMS) if, for every pair of agents $i$ and $j$, it holds that $v_i(X_i) \ge \mu_i(X_i \cup X_j)$. Otherwise, if $v_i(X_i) < \mu_i(X_i \cup X_j)$, we say that agent $i$ \emph{PMMS-envies} agent $j$.
\end{definition}

We also study the related concept of MMS, and compare it to PMMS in \Cref{sec:nonexistence}.

\begin{definition}[MMS Allocations]
\label{def:mms}
An allocation $X$ is \emph{maximin share fair} (MMS) if, for every agent $i$, it holds that $v_i(X_i) \ge \mu_i(M, n)$.
\end{definition}

\subsection{Valuation Classes}

An important class of valuations is the class of \emph{additive} valuations, where for every $S \subseteq M$ and any $g \in M \setminus S$, we have $v(S \cup \{g\}) = v(S) + v(\{g\})$.

We now formally define personalized bivalued valuations.

\begin{definition}[Personalized Bivalued Valuations]
\label{def:personalizedbivalued}
Agent $i$'s valuation is said to be \emph{personalized bivalued} if $v_i$ is additive and there exist values $a_i > b_i \ge 0$ such that $v_i(\{g\}) \in \{a_i, b_i\}$ for every item $g \in M$. Furthermore, a personalized bivalued valuation is called \emph{factored} if either $a_i$ is divisible by $b_i$, or $b_i = 0$.
\end{definition}

Whenever we write $a_i / b_i$  {and $b_i = 0$, we interpret it as a sufficiently large number $K$}.

We also define the class of MMS-feasible valuations, following the definition of \citet*{DBLP:journals/ior/AkramiACGMM25}.

\begin{definition}[MMS-Feasible Valuations] \label{def:mmsfeasibility}
    Agent $i$'s valuation is \emph{MMS-feasible} if, for every subset $S \subseteq M$ and for any two partitions $(X_1, X_2)$ and $(Y_1, Y_2)$ of $S$ such that $X_1 \cup X_2 = Y_1 \cup Y_2 = S$ and $X_1 \cap X_2 = Y_1 \cap Y_2 = \emptyset$, it holds that $\max\{v_i(X_1), v_i(X_2)\} \ge \min\{v_i(Y_1), v_i(Y_2)\}$.
\end{definition}

We define binary-valued valuations as follows. 

\begin{definition}[Binary-Valued Valuations]\label{def:binary}
Agent $i$'s valuation is \emph{binary} if for every subset $S \subseteq M$, it holds that $v_i(S) \in \{0, 1\}$.
\end{definition}

Finally, we define the class of pair-demand valuations.

\begin{definition}[Pair-Demand Valuations]\label{def:pairdemand}
Agent $i$'s valuation is called \emph{pair-demand} if there exist non-negative values $v_{i,1}, \ldots, v_{i,m} \geq 0$ such that, for every subset $S \subseteq M$, it holds that:
\begin{align*}
    v_i(S) = \max_{T \subseteq S, |T| \leq 2} \sum_{j \in T} v_{i,j}.
\end{align*}
\end{definition}

\subsection{Implications of PMMS}\label{sec:implicationsofpmms}

In this section, we briefly examine the relationship between PMMS and EFX. While PMMS does not imply EFX in general, it does imply a slightly weaker variant of EFX, as we discuss below. Moreover, in all non-degenerate instances, PMMS implies the full EFX condition.

The example below demonstrates that PMMS does not necessarily imply EFX in general.

\begin{example}[PMMS Does Not Imply EFX]\label{ex:pmmsdoesnotimplyefx}
    Consider two agents and three items. Both agents have identical additive valuations $v$, where $v(\{1\}) = v(\{2\}) = 0$ and $v(\{3\}) = 2$. Allocating the bundle $\{1\}$ to the first agent and $\{2,3\}$ to the second agent satisfies PMMS. However, it is not EFX, as the first agent would still envy the second even after the removal of item $2$.
\end{example}

The key issue in the example above is the presence of zero-valued items. Consequently, in the context of additive valuations over goods, many works have studied a related notion called envy-freeness up to any \emph{positively valued} good (EFX$_{>0}$), which requires that for every pair of agents $i$ and $j$, we have $v_i(X_i) \ge v_i(X_j \setminus \{g\})$ for any $g \in X_j$ such that $v_i(\{g\}) > 0$. It is known that PMMS implies EFX$_{>0}$ for any additive valuation \cite{unreasonable-fariness}. Furthermore, the definition of EFX$_{>0}$ has been extended to non-additive and non-monotone valuations by \citet*{abs-2502-02815}, where the implication from PMMS to EFX$_{>0}$ continues to hold.

In addition, \Cref{ex:pmmsdoesnotimplyefx} is a degenerate instance in the sense that agents assign equal value to multiple distinct bundles. Notably, \citet*{efx-3-agents} observe that establishing the existence of EFX allocations for all non-degenerate instances implies their existence for all instances in general. Therefore, resolving the EFX problem for non-degenerate instances would settle the general EFX problem.
Crucially, in non-degenerate instances, PMMS implies EFX.

\section{Non-Existence of PMMS}\label{sec:nonexistence}

In this section, we examine the conditions necessary for the existence of PMMS allocations. We begin by discussing the two-agent case and then proceed to the non-trivial case of three agents.

In the case of two agents, the PMMS condition is equivalent to MMS (\Cref{def:mms}). It has been noted by \citet*{DBLP:journals/ior/AkramiACGMM25} that MMS allocations for two agent exist when one agent has an MMS-feasible valuation and the other has an arbitrary valuation. This condition is necessary, as an MMS allocation may not exist for two agents if both valuations are arbitrary.

While one might expect that a non-existence result for two agents would automatically extend to settings with three or more agents, this extension is not immediate.
For example, simply adding a dummy agent who values every subset at $0$ to the two-agent counterexample is not sufficient.
Intuitively, adding such a dummy agent is equivalent to allowing some items to go to charity, as considered by \citet*{DBLP:conf/ec/CaragiannisGH19}.
We present a non-trivial extension of the non-existence result beyond the two-agent case.
To the best of our knowledge, this is the first known construction demonstrating the non-existence of a PMMS allocation for more than two agents.

\begin{proposition}[Non-Existence of PMMS]\label{prop:pmmsnonexistance}
   For any $n \geq 2$, there exists an instance with $n$ agents, each with monotone valuations, and $n + \mathcal{O}(\log n)$ items, in which no PMMS allocation exists.
Moreover, an MMS allocation exists in this instance.
\end{proposition}

\begin{proof}
   Let $n \ge 2$, and choose $k$ such that $\binom{2k}{k} \ge 2n$. We construct an instance with $m = 2k + n - 2$ items. Note that it suffices to take $k = \mathcal{O}(\log n)$, and therefore $m = n + \mathcal{O}(\log n)$.
We refer to the first $n - 2$ items as \emph{stars}, denoted by $s_1, \ldots, s_{n-2}$, and the remaining $2k$ items as \emph{common} items, denoted by $c_1, \ldots, c_{2k}$.
    
    For each agent $i$, we fix a partition $(A_i, B_i)$ of the common items ${c_1, \ldots, c_{2k}}$ such that the partitions are unique across agents. More formally, for agent $i$, we require that $|A_i| = |B_i| = k$, $A_i \cap B_i = \emptyset$, and for every pair of distinct agents $i$ and $j$, it holds that $A_i \notin \{A_j, B_j\}$.
This is always possible because we have chosen $2k$ such that $\binom{2k}{k} \geq 2n$.
    
    We then set the valuation $v_i$ so that:
    \begin{enumerate}[noitemsep, topsep=0pt,label={\textit{(\roman*)}}]
        \item Each star item has value $k$; that is, $v_i(\{s_j\}) = k$ for all $j \in \{1, \ldots, n-2\}$.
        \item The special bundles $A_i$ and $B_i$ have value $k+1$; that is, $v_i(A_i) = v_i(B_i) = k +1$.
         \item Any bundle containing at least one star and at least one additional item has value $2k$; that is, $v_i(X) = 2k$ for any $X$ such that $|X| \geq 2$ and $s_j \in X$ for some $j \in \{1, \ldots, n-2\}$.
        \item Any other bundle $X$ consisting only of common items has value equal to its cardinality, i.e., $v_i(X) = |X|$ for all $X \subseteq \{c_1, \ldots, c_{2k}\}$ and $X \notin \{A_i, B_i\}$.
    \end{enumerate}

    We now show that no PMMS allocation exists for this instance. Fix an arbitrary allocation $X = (X_1, \ldots, X_n)$, and consider the following cases.

    \textit{Case 1: Some agent $i$ gets an empty bundle.}
    This violates the PMMS condition, since the agent with the empty bundle would envy any other agent who receives at least two items. Given that the total number of items is $2k + n - 2 \geq n$, there must be at least one such agent.
Any bundle containing two or more items can be partitioned into two non-empty bundles, each with positive value, which gives agent $i$ a higher fair share value than zero.
    
    \textit{Case 2: Some agent $i$ gets at least one star and at least one commmon item.}
Since there are only $n - 2$ stars and $n$ agents, at least two agents, say $j$ and $r$, receive no stars.
Additionally, since there are only $2k$ common items and one of them is allocated to agent $i$, either $j$ or $r$ must receive at most $k - 1$ common items.
Assume without loss of generality that agent $j$ receives no star and at most $k - 1$ common items.
Note that $v_j(X_j) = |X_j|$.
Agent $j$ would PMMS-envy agent $i$, since $X_i \cup X_j$ can be partitioned into two bundles: one containing the star (valued at $k$) and the other containing $|X_j| + 1$ common items (valued at $|X_j| + 1$).
Both bundles would be strictly more valuable than $v_j(X_j)$, implying that agent $j$ PMMS-envies agent $i$.

    \textit{Case 3: Some agent $i$ gets at least two stars.}
    Since there are only $n - 2$ stars in total and agent $i$ receives two of them, at least three agents must receive no stars. With only $2k$ common items to distribute among these three agents, at least one of them, say, agent $j$, must receive at most $k - 1$ common items. Note that $v_j(X_j) = |X_j|$. Agent $j$ would then PMMS-envy agent $i$, since the combined bundle $X_i \cup X_j$ can be partitioned into two bundles, each containing a star (valued at least $k$). Both bundles would have a value strictly greater than $v_j(X_j)$, implying that agent $j$ PMMS-envies agent $i$.
      
    \textit{Case 4: All agents who receive a star receive no common items.} In this scenario, exactly $n - 2$ agents each receive a single star and nothing else, while the remaining two agents, say $i$ and $j$, divide all $2k$ common items between themselves.
Note that $\mu_i(X_i \cup X_j) = \mu_i({c_1, \ldots, c_{2k}}) = \max\{v_i(A_i), v_i(B_i)\} = k + 1$, and similarly, $\mu_j(X_i \cup X_j) = k + 1$.
However, since the partitions $(A_i, B_i)$ and $(A_j, B_j)$ are unique, it must be that either $X_i \notin \{A_i, B_i\}$ or $X_j \notin \{A_j, B_j\}$.
In the first case, we have $v_i(X_i) < \mu_i(X_i \cup X_j)$; in the second case, $v_j(X_j) < \mu_j(X_i \cup X_j)$.
In either case, the PMMS condition is violated.

   We can now conclude that there is no PMMS allocation in the constructed instance.

 Moreover, we show that for each agent $i$, the MMS share is $\mu_i(M, n) = k$.
To establish this, it suffices to argue that there is no partition of the items into $n$ bundles such that every bundle is valued at $k + 1$ by agent $i$.
Indeed, since there are only $n - 2$ stars, two of the bundles must contain no stars.
The only way for those two bundles to each be valued at $k + 1$ is if the common items are split according to $(A_i, B_i)$.
However, in that case, the remaining $n - 2$ bundles, each containing a single star, are valued exactly $k$, which proves the claim.

We can guarantee that every agent receives a bundle valued at least $k$ by allocating the $n - 2$ stars to $n - 2$ agents and dividing the common items equally between the two remaining agents.
Therefore, an MMS allocation exists.
\end{proof}

The construction in the proof of \Cref{prop:pmmsnonexistance} uses $n + \mathcal{O}(\log n)$ items.
An interesting open problem is to determine the minimal number of items required for such a construction; specifically, whether the result can be improved to include only $n + \mathcal{O}(1)$ items; see \Cref{sec:conclusion}.

\input{figures/simplethreeagents}

As shown by \citet*{DBLP:journals/ior/AkramiACGMM25}, an EFX allocation always exists for three agents if at least one of them has an MMS-feasible valuation and the other two have monotone valuations.
We show that these assumptions are not sufficient to guarantee the existence of a PMMS allocation.
Specifically, we present an example with two agents having monotone valuations and one agent with an additive (hence MMS-feasible) valuation, where no PMMS allocation exists.
This establishes a separation between the two fairness notions and suggests that the proof technique used to establish the existence of EFX allocations for three agents does not extend to PMMS.
The proof of the following theorem is deferred to \Cref{sec:proofsnonexistence}.

\thmseparation*

The instance used in the proof of \Cref{thm:separation} was discovered through a computational search over the space of possible valuations.
Considerable effort was devoted to simplifying the counterexample, as the initial instances had complex and unintuitive structure.
The counterexample presented in this work is notably simpler, involving only few distinct values and handling cases with uneven bundle sizes in an elegant way.

\section{Personalized Bivalued Valuations}
\label{sec:bivalued}

In this section, we analyze the case of personalized bivalued valuations, where the valuations are additive and, for each agent $i$, there exist some $a_i > b_i \geq 0$ such that $v_i(\{g\}) \in \{a_i, b_i\}$ for all items $g$.
Our main theorem for this setting is as follows.

\thmbivalued*

\input{algos/bivaluedvaluations}

The remainder of this section is devoted to proving \Cref{thm:bivalued}. Specifically, we show that \Cref{alg:personalized-bivalued-pmms} produces an EFX allocation for any personalized bivalued instance, and a PMMS allocation if the instance is factored. \Cref{alg:personalized-bivalued-pmms} is a modified version of the Match-and-Freeze algorithm proposed by \citet*{Amanatidis_2021}, which was designed to compute EFX allocations for non-personalized bivalued valuations.

We begin by proving the following property of the allocation produced by the algorithm.

\begin{lemma}
    \label{lem:max-le-min}
     Consider an execution of \Cref{alg:personalized-bivalued-pmms}, and fix a round $r$. Let $C$ be any connected component of the graph $G$ constructed during round $r$. If the set $U$ of unmatched agents in $C$ is nonempty and $U \neq C$, then it holds that:
    \begin{align*}
    \max\{a_i/b_i \colon i\in U\} \le \min\{a_i/b_i \colon i\in C\setminus U\}.
    \end{align*}
\end{lemma}

\begin{proof}
Fix a round $r$ and a component $C$ of the graph $G$, with a set of unmatched agents $U$.
Suppose, for contradiction, that there exists an unmatched agent $j \in U$ and a matched agent $i \in C \setminus U$ such that $a_i/b_i < a_j/b_j$.
Since $i$ and $j$ belong to the same component, there exists an alternating path $i_0, g_1, i_1, g_2, i_2, \ldots, g_k, i_k$ in $G$, where $i_0 = j$, $i_k = i$, and $(i_w, g_{w+1}) \in G$ for all $0 \leq w < k$, and $(i_w, g_w) \in G$ for all $0 < w \leq k$.
Along this path, agents $i_1, i_2, \ldots, i_k$ are matched to items $g_1, g_2, \ldots, g_k$, respectively. See \Cref{fig:alternating-path-contradiction} for an illustration of the alternating path used in the argument.

Notice that we can increase the weight of the maximal matching by reassigning items $g_1, g_2, \ldots, g_k$ to agents $i_0, i_1, \ldots, i_{k-1}$, respectively. This increases the total weight because the weight of each edge $(i_w, g_w)$ is equal to the weight of the edge $(i_w, g_{w+1})$ for all $0 < w < k$, and the weight of the edge $(i_0, g_1)$ is greater than that of the edge $(i_k, g_k)$.
This contradicts the choice of the matching in \Cref{line:matching}, completing the proof of the lemma.
\end{proof}

Next, we use \Cref{lem:max-le-min} to derive the following technical lemma, which plays a key role in our analysis. The lemma and its proof are similar to the analysis of the Match-and-Freeze algorithm for non-personalized bivalued valuations, as presented by \citet*[Lemma 4.2]{Amanatidis_2021}. The proof is deferred to \Cref{sec:proofsfrombivalued}.

\begin{restatable}{lemma}{lemprincipallemma}
    Consider an execution of \Cref{alg:personalized-bivalued-pmms}, and fix an agent $i$. Let $r_i$ denote the last round in which an item $g$ with $v_i(\{g\}) = a_i$ was allocated to some agent. Let $X_{i,r}$ denote the singleton set allocated to agent $i$ in round $r$, or the empty set if agent $i$ does not receive any item in that round. Let $F_i$ denote the set of rounds during which agent $i$ is not active. Then:
\begin{enumerate}[label=(\roman*)]
\item We have $v_i(X_{i,r}) = a_i$ for all $r \in \{1, 2, \ldots, r_i - 1\}$.
\item If $F_i \neq \emptyset$, then $v_i(X_{i,r_i}) = a_i$ and $F_i \subseteq \{r_i+1, \ldots, r_i + \lfloor a_i/b_i -1\rfloor\}$.
\item If $v_i(X_{i,r_i}) = a_i$ and $v_i(X_{j,r_i}) = a_i$ for some agent $j$, then $F_i = F_j$.
\end{enumerate}
    \label{lem:principal-lemma}
\end{restatable}

\input{figures/alternating_path}

        The proof of \Cref{thm:bivalued} also relies on the following technical claim.

\begin{restatable}{lemma}{lemsufficcc}\label{lem:sufficcc}
Let $X$ be an allocation, and let $i$ and $j$ be any agents. Suppose that $v_i(X_i) \geq v_i(X_j) - b_i$. Then agent~$i$ does not EFX-envy agent~$j$. Moreover, if agent~$i$'s valuation is factored, then agent~$i$ also does not PMMS-envy agent~$j$.
\end{restatable}
\begin{proof}
    The EFX condition follows immediately from this inequality, since any item is worth at least $b_i$ to agent $i$. For the PMMS condition, assume that $i$'s valuation is factored. Without loss of generality, we may assume $b_i \in \{0, 1\}$ and that $a_i$ is a non-negative integer. We obtain:
\begin{align*}
    v_i(X_i) \ge (v_i(X_i) + v_i(X_j) - b_i)/2 \ge \left\lfloor v_i(X_i \cup X_j)/2 \right\rfloor \ge \mu_i(X_i \cup X_j),
\end{align*}
where the last inequality follows from the fact that $\mu_i(X_i \cup X_j)$ must be an integer. This proves the PMMS part of the claim.
\end{proof}

We are now ready to prove the main theorem.

\begin{proof}[Proof of \Cref{thm:bivalued}]
Consider an execution of \Cref{alg:personalized-bivalued-pmms}, and fix two agents $i$ and $j$.  
Let $R$ denote the index of the last round, and let $r_i \leq R$ denote the last round in which an item $g$ with $v_i(\{g\}) = a_i$ was allocated to any agent. Let $X_{i,r}$ denote the singleton set allocated to agent $i$ in round $r$, or the empty set if agent $i$ does not receive any item in that round. Similarly, let $X_{j,r}$ denote the set allocated to agent $j$ in round $r$. Let $F_i \subseteq \{1, \ldots, R\}$ denote the set of rounds during which agent $i$ is inactive.

Our goal is to prove that agent $i$ does not EFX-envy agent $j$, and if the valuation is factored, then $i$ also does not PMMS-envy $j$. By \Cref{lem:sufficcc}, it suffices to show that $v_i(X_i) \geq v_i(X_j) - b_i$. We use this observation to complete the proof in most of the cases considered in our analysis.

Next, by \Cref{lem:principal-lemma}, we know that $v_i(X_{i,r}) = a_i$ for all $r \in \{1, \ldots, r_i - 1\}$, and $v_i(X_{i,r}) \geq b_i$ for all $r \in \{r_i + 1, \ldots, R - 1\} \setminus F_i$, where $|F_i| \leq \lfloor a_i /b_i - 1 \rfloor$. Moreover, $v_i(X_{j,r}) \leq a_i$ for all $r \in \{1, \ldots, r_i - 1\}$ and $v_i(X_{j,r}) \leq b_i$ for all $r \in \{r_i + 1, \ldots, R\}$.

\textit{Case 1: $v_i(X_{i,r_i}) = a_i$ and $v_i(X_{j,r_i}) = b_i$.}  
In this case, we have:
\begin{align*}
    v_i(X_i) &\geq r_i \cdot a_i + (R - 1 - r_i - \lfloor a_i / b_i - 1 \rfloor) \cdot b_i \\
    &= r_i \cdot a_i + (R - r_i) \cdot b_i - (\lfloor a_i / b_i - 1 \rfloor + 1) \cdot b_i\\
    &\geq r_i \cdot a_i + (R - r_i) \cdot b_i - ((a_i/b_i - 1) + 1) \cdot b_i \\
&= (r_i - 1) \cdot a_i + (R - r_i) \cdot b_i \\
&\geq v_i(X_j) - b_i.
\end{align*}

\textit{Case 2: $v_i(X_{i,r_i}) = a_i$ and $v_i(X_{j,r_i}) = a_i$.}  
In this case, by \Cref{lem:principal-lemma}, we have $F_i = F_j$, and:
\[
v_i(X_i) \geq r_i \cdot a_i + (R - 1 - r_i - |F_i|) \cdot b_i 
 \geq v_i(X_j) - b_i.
\]

\textit{Case 3: $v_i(X_{i,r_i}) = b_i$ and $v_i(X_{j,r_i}) = b_i$.}  
In this case, by \Cref{lem:principal-lemma}, we have $F_i = \emptyset$, and:
\[
v_i(X_i) \geq (r_i - 1) \cdot a_i + (R - r_i) \cdot b_i \geq v_i(X_j) - b_i.
\]

\textit{Case 4: $v_i(X_{i,r_i}) = b_i$ and $v_i(X_{j,r_i}) = a_i$.}  
In this case, by \Cref{lem:principal-lemma}, we know that agent $i$ was not matched during round $r_i$, and thus $F_i = \emptyset$. Moreover, agent $j$ must have been matched to item $X_{j,r_i}$ during round $r_i$. Indeed, it cannot be that agent $j$ received item $X_{j,r_i}$ through the operation in \Cref{line:remitems}, because if that were the case, then agent $i$ would have been matched to item $X_{j,r_i}$ during round $r_i$ by the construction of the algorithm.
Therefore, during round $r_i$, agents $i$ and $j$ belong to the same connected component, as they both have an edge to item $X_{j,r_i}$. Furthermore, by the choice of freezing time in \Cref{lin:freezingtime}, we have $v_i(X_{j,r}) = 0$ for all $r \in \{r_i + 1, \ldots, r_i + \lfloor a_i / b_i - 1 \rfloor\}$, since agent $i$ remains unmatched during round $r_i$.

First, suppose that $R > r_i + \lfloor a_i / b_i - 1 \rfloor$.  
Since agent $j$ becomes frozen, her priority is updated in \Cref{lin:settingpriority}, so that agent $i$ always chooses an item before agent $j$. Hence, $v_i(X_{i,R}) \geq v_i(X_{j,R})$. It follows that:
\begin{align*}
v_i(X_i) &= (r_i - 1) \cdot a_i + (R - r_i) \cdot b_i + v_i(X_{i,R}) \\
&= r_i \cdot a_i + (R - a_i/b_i - r_i) \cdot b_i + v_i(X_{i,R}) \\ 
&\geq r_i \cdot a_i + (R - 2 - \lfloor a_i / b_i - 1 \rfloor - r_i) \cdot b_i + v_i(X_{j,R}) \\
&\geq v_i(X_j) - b_i.
\end{align*}

Now, suppose that $R \leq r_i + \lfloor a_i / b_i - 1 \rfloor$.  
In this case, $X_j$ consists of exactly $r_i$ items, and $X_i$ consists of at least $r_i - 1$ high-value items and possibly some low-value items. The EFX condition holds since $v_i(X_i) \geq (r_i - 1) \cdot a_i \geq v_i(X_j \setminus \{g\})$ for any $g \in X_j$.
For the PMMS condition, observe that the total value of low-value items in $X_i$ is at most $(\lfloor a_i / b_i - 1 \rfloor + 1) \cdot b_i \leq a_i$, and therefore:
\[
\mu_i(X_i \cup X_j) \leq (r_i - 1) \cdot a_i + (|X_i| - r_i + 1) \cdot b_i = v_i(X_i).
\]

This completes the proof of the theorem.
\end{proof}

In the following remark, we explain the importance of the order in which unmatched agents receive items in \Cref{lin:unmatched_order} of \Cref{alg:personalized-bivalued-pmms}.

\begin{remark}
Consider the example described in \Cref{tab:allocation-rounds}. In the final round of the algorithm's execution, only two items remain, and all agents are active. Therefore, two agents must be left without an item in this round. According to the priority ordering $w$, the items are given to agents $1$ and $3$, as they were not matched at any point during the execution of the algorithm.

If, instead, an item were allocated to agent $2$ rather than agent $1$, then agent $1$ would EFX-envy agent $2$. In that case, agent $1$ would receive five low-value items, yielding a value of $5$, while agent $2$ would receive one high-value item and four low-value items. Agent $1$'s valuation for agent $2$'s bundle, even after removing one low-value item, would be $5.5$, violating EFX.
\end{remark}

We also observe that \Cref{alg:personalized-bivalued-pmms} allocates at least one item in each round, ensuring that it runs in polynomial time (in $n$ and $m$), as a maximum-weight maximal matching can be computed in polynomial time \cite{kuhn1955hungarian}. This leads to the following corollary.

\begin{corollary}\label{cor:compute_bivalued}
For any instance with personalized bivalued valuations, an EFX allocation can be computed in polynomial time. Furthermore, for any instance with factored personalized bivalued valuations, a PMMS allocation can also be computed in polynomial time.
\end{corollary}

Finally, we note that \citet*{Amanatidis_2021} also showed that for non-personalized bivalued valuations, an EFX allocation can be achieved by maximizing the Nash welfare (i.e., the geometric mean of agents' utilities, $(\prod_{i \in N} v_i(X_i))^{1/n}$). In sharp contrast, we show that this result does not extend to personalized bivalued valuations, as demonstrated by the following proposition.

\begin{proposition}
    \label{prop:mnw-fail-personalized} 
    There exists an instance with personalized bivalued valuations where none of the allocations maximizing Nash welfare satisfies EFX.
\end{proposition}
\begin{proof}
We construct an instance with two agents and four items. The personalized bivalued valuation for the first agent has parameters $a_1 = 5$ and $b_1 = 1$, and for the second agent, $a_2 = 3$ and $b_2 = 1$. For both $i \in \{1, 2\}$, we set the valuations to satisfy $v_i(\{g_1\}) = v_i(\{g_2\}) = a_i$ and $v_i(\{g_3\}) = v_i(\{g_4\}) = b_i$. The maximum Nash welfare for this instance is equal to $25$, which is achieved by two allocations: $\langle \{g_1\}, \{g_2, g_3, g_4\} \rangle$ and $\langle \{g_2\}, \{g_1, g_3, g_4\} \rangle$.
In both of these allocations, the first agent EFX-envies the second agent, which completes the proof of the proposition.
\end{proof}

\section{Binary-Valued Valuations}
\label{sec:binary}

In this section, we consider the class of binary-valued valuations, where all possible bundles are classified as either desirable or non-desirable; see \Cref{def:binary} for a formal definition. Importantly, we do not assume that these valuations are monotone or normalized, we only assume that they satisfy the MMS-feasibility condition (\Cref{def:mmsfeasibility}).
We now proceed to state our main theorem.

\thmbinary*

\input{algos/binaryvaluations}

The proof of \Cref{thm:binary} is based on the notion of the cut-and-choose graph, which is formally defined in the following definition. This notion is inspired by the envy graph introduced by \citet*{lipton}.

\begin{definition}[Cut-and-Choose Graph]
    For a fixed allocation $(X_1, \ldots, X_n)$ and a fixed agent $s \in N$, we define the cut-and-choose graph as follows.
The set of vertices is the set of agents $N$.
For every agent $i \in N$, we add an edge to precisely one agent $\pi(i) \in N$.
If $v_i(X_s) \geq \pmms_i(X_s \cup X_j)$ for all $j \in N$, then we set $\pi(i) = s$.
Otherwise, we set $\pi(i) = j$ for some $j \in N$ such that $v_i(X_s) < \pmms_i(X_s \cup X_j)$.
    \label{def:cutandchoosegraph}
\end{definition}

We use the following key properties of the cut-and-choose graph.

\begin{lemma}\label{lem:cutandchooseproperties}
Fix an allocation $(X_1, \ldots, X_n)$ and an agent $s \in N$, and consider the cut-and-choose graph. Then, for every agent $i \in N$, the following holds:
\begin{enumerate}[label=(\roman*)]
    \item If $\pi(i) = s$, then either $v_i(X_s) = 1$, or $\pmms_i(X_s \cup X_j) = 0$ for all $j \in N$. \label{enum:cutandchooseproperties_one}
    \item If $\pi(i) = j$ for some $j \neq s$, then $v_i(X_j) = 1$ and $\pmms_i(X_s \cup X_j) = 1$.\label{enum:cutandchooseproperties_two}
\end{enumerate}
\end{lemma}

\begin{proof}
(i) Suppose $\pi(i) = s$. By the definition of the cut-and-choose graph (\Cref{def:cutandchoosegraph}), this implies that for all $j \in N$, we have
$$
v_i(X_s) \geq \pmms_i(X_s \cup X_j).
$$
This inequality holds if and only if either $v_i(X_s) = 1$, or $\pmms_i(X_s \cup X_j) = 0$ for all $j \in N$.

(ii) Suppose $\pi(i) = j$ for some $j \neq s$. Then, again by the definition of the cut-and-choose graph (\Cref{def:cutandchoosegraph}), it must be that
$$
v_i(X_s) < \pmms_i(X_s \cup X_j).
$$
Since $v_i(X_s) \in \{0, 1\}$ and $\pmms_i(X_s \cup X_j) \in \{0, 1\}$, the inequality implies $v_i(X_s) = 0$ and $\pmms_i(X_s \cup X_j) = 1$. 
By the definition of MMS-feasibility (\Cref{def:mmsfeasibility}), we have
    \[
    \max\left\{v_i(X_s), v(X_j)\right\} \geq \pmms_i(X_s \cup X_j) = 1.
    \]
    Since $v_i(X_s) = 0$, it must be that $v_i(X_j) = 1$.
\end{proof}

We now use \Cref{lem:cutandchooseproperties} to prove \Cref{thm:binary}. The proof relies on \Cref{alg:binary-pmms}. See \Cref{fig:binaryvaluations} for an illustration of the operations performed by the algorithm.

\begin{proof}[Proof of \Cref{thm:binary}]
By the while-loop condition in \Cref{line:binary_while}, it is clear that the output of \Cref{alg:binary-pmms} satisfies the PMMS condition. Hence, to prove \Cref{thm:binary}, it suffices to show that \Cref{alg:binary-pmms} terminates for any input.

Consider the tuple $\langle W, E \rangle$, where $W = \sum_{i \in N} v_i(X_i)$ and $E$ is the number of agents for whom the PMMS condition holds. 
We will show that in each iteration, either $W$ strictly increases, or $W$ remains the same and $E$ strictly increases. 
Since $W \in \{0, 1, \ldots, n\}$ and $E \in \{0, 1, \ldots, n\}$, the algorithm must terminate after at most $n^2$ iterations.

To prove this claim, fix an iteration of the while-loop. Let $X$ be the allocation at the beginning of the iteration and let $Y$ be the allocation at the end.

\emph{Case 1: Cycle.} Suppose the if-condition in \Cref{line:binary_if} is satisfied. By \Cref{lem:cutandchooseproperties}\ref{enum:cutandchooseproperties_two}, for all $i \in \{i_0, \ldots, i_{k-1}\}$, it holds that $v_i(Y_i) = v_i(X_{\pi(i)}) = 1$. Moreover, by \Cref{lem:cutandchooseproperties}\ref{enum:cutandchooseproperties_one}, we know that $Y_{i_k} = X_{i_0}$, and either (i) $v_{i_k}(X_{i_0}) = 1$, or (ii) $\pmms_{i_k}(X_{i_0} \cup X_j) = 0$ for all $j \in N$. 

Observe that initially, $\sum_{h=0}^k v_{i_h}(X_{i_h}) < k+1$ since $v_{i_0}(X_{i_0}) = 0$. Furthermore, the number of agents in $\{i_0, \ldots, i_k\}$ for whom the PMMS condition holds is strictly less than $k+1$, since agent $i_0 = s$ does not satisfy the PMMS condition by assumption in \Cref{line:binarychoiceofs}.

In case (i), $W$ strictly increases since after the reassignment, $\sum_{h=0}^k v_{i_h}(Y_{i_h}) = k+1$. In case (ii), although $v_{i_k}(Y_{i_k}) = v_{i_k}(X_{i_0})$ may be 0, all other agents in the cycle gain value 1, so $\sum_{h=0}^k v_{i_h}(Y_{i_h}) \geq k$. Thus, $W$ does not decrease. Moreover, the PMMS condition must hold for all agents in $\{i_0, \ldots, i_k\}$ since $\pmms_{i_k}(X_{i_0} \cup X_j) = 0$ for all $j \in N$. Hence, $E$ increases.

\emph{Case 2: Lollipop.} Now suppose the if-condition in \Cref{line:binary_if} is not satisfied. By \Cref{lem:cutandchooseproperties}\ref{enum:cutandchooseproperties_two}, we have $v_i(Y_i) = v_i(X_{\pi(i)}) = 1$ for all $i \in \{i_0, \ldots, i_k\} \setminus \{i_{w-1}, i_k\}$. Also, from the same lemma, we know that $\pmms_{i_{w-1}}(X_{i_w} \cup X_{i_0}) = 1$ and $\pmms_{i_k}(X_{i_w} \cup X_{i_0}) = 1$.

Since $(A, B)$ is set to be a maximin partition of $X_{i_w} \cup X_{i_0}$ by $v_{i_k}$ in \Cref{line:binarychoiceofaandb}, we have $v_{i_k}(B) = 1$ since $\pmms_{i_k}(X_{i_w} \cup X_{i_0}) = 1$. Finally, by the swap operation in \Cref{line:binaryswappingofaandb} and the MMS-feasibility of $v_{i_{w-1}}$, we have
$$
v_{i_{w-1}}(Y_{i_{w-1}}) = v_{i_{w-1}}(A) = \max \{v_{i_{w-1}}(A), v_{i_{w-1}}(B)\} \geq \pmms_{i_{w-1}}(A \cup B) = 1.
$$
Therefore, $W$ strictly increases in this case as well. This proves the claim, and consequently, completes the proof of the theorem.
\end{proof}

As noted in the proof of \Cref{thm:binary}, the tuple $\langle W, E \rangle$ increases lexicographically with every iteration of \Cref{alg:binary-pmms}. Since there are only $n^2$ possible values for this tuple, \Cref{alg:binary-pmms} runs in polynomial time (in $n$ and $m$). This leads to the following corollary.

\begin{corollary}\label{cor:compute_binary}
For any instance with binary-valued valuations, a PMMS allocation can be computed in polynomial time.
\end{corollary}

\section{Pair-Demand Valuations}\label{sec:pairdemand}

In this section, we analyze the natural class of pair-demand valuations, where each agent desires at most two items; see \Cref{def:pairdemand} for a formal definition. This class generalizes unit-demand valuations, in which every agent wants at most one item. While fairness notions such as EFX and PMMS are straightforward to guarantee under unit-demand valuations, they become non-trivial even in the pair-demand setting.
Our main result in this section is the following theorem.

\thmpmmspairdemand*

We first remark that our proof does not establish the existence of EFX allocations for pair-demand valuations, which we leave as an open problem; see \Cref{sec:conclusion}. In particular, PMMS does not imply EFX, as valuations may be degenerate in the sense discussed in \Cref{sec:implicationsofpmms}. Achieving EFX for pair-demand valuations may require allocating only a single item to one agent and more than two items to another. For example, consider an instance with two identical agents, one high-value item, and three low-value items. In this case, one agent must receive the high-value item, while the other receives all three low-value items, even though the third low-value item no longer increases the second agent's value and could instead benefit the first agent.
As we show below, it is possible to obtain a PMMS allocation while giving each agent exactly two items, provided there are at least $2n$ items in total.

Our proof is based on a two-stage round-robin algorithm (\Cref{alg:pair-demand-pmms}) in which agents select items according to a fixed order in the first stage, and then in the reversed order in the second stage. The idea of reversing the order in round-robin algorithms has been explored in the literature; for instance, in ABBA picking sequences by \citet*{brams2000win}, the Double Round-Robin algorithm of \citet*{AzizCIW22}, and as part of the Draft-and-Eliminate algorithm proposed by \citet*{Amanatidis_2020}.

\input{algos/pmmspairdemand.tex}

We being our analysis by establishing the following lemma.

\begin{lemma}
  \label{lem:additive-pair-demand-envy}
  Consider an agent $i$ with a pair-demand valuation $v_i$. Let $M = \{a, b, c, d\}$, and suppose that $v_i(\{a\}) \le v_i(\{b\}) \le v_i(\{c\}) \le v_i(\{d\})$.
Then, for any subset $S \subseteq M$ with $|S| = 2$, if $S \notin \left\{ \{a,b\}, \{a,c\} \right\}$, it holds that $v_i(S) \ge \mu_i(M)$.
  \end{lemma}

\begin{proof}
    To calculate the PMMS value, we only need to consider three partitions of $M$ into two 2-item bundles. The fair share of agent $i$ is given by:
  \begin{align*}
    \mu_i(M) = \max\{& \min \{v(\{a, b\}), v(\{c, d\})\},\\
               & \min\{v(\{a, c\}), v(\{b, d\})\},\\
               & \min\{v(\{a, d\}), v(\{b, c\})\} \, \}
  \end{align*}
  By our assumption on the ordering of the items, we get that  $\min \{v(\{a, b\}), v(\{c, d\})\} = v(\{a,b\})$ and $\min\{v(\{a, c\}), v(\{b, d\})\} = v(\{a, c\})$. Moreover, since 
  \begin{align*}
    v(\{a, b\}) \le v(\{a, c\}) \le
  \min\left\{v(\{a, d\}), v(\{b, c\})\right\},      
  \end{align*}
it follows that $\mu_i(M) = \min\{v(\{a, d\}), v(\{b,
  c\})\}$. The result follows.
\end{proof}

We now prove the main theorem.

\begin{proof}[Proof of \Cref{thm:pmms-pair-demand}]
We will prove that \Cref{alg:pair-demand-pmms} returns a PMMS allocation.
By adding dummy items of value zero, we may assume  that there are at least $2n$ items.

Fix any two agents $i$ and $j$. Let $g_i$ and $h_i$ be the items allocated to agent $i$ in the first and second rounds, respectively. Let $x_j$ and $y_j$ be the two most valuable items in $X_j$ according to $v_i$, with $v_i(\{x_j\}) \ge v_i(\{y_j\})$.

    First, suppose that $i < j$. In this case, agent $i$ picked the first item before agent $j$ picked any item, so $v_i(\{g_i\}) \ge v_i(\{x_j\})$ and $v_i(\{g_i\}) \ge v_i(\{y_j\})$. By \Cref{lem:additive-pair-demand-envy}, it follows that $v_i(\{g_i,h_i\}) \geq \mu_i(\{g_i,h_i, x_j, y_j\}) \geq \mu_i(X_i \cup X_j)$.

   Next, suppose that $i > j$. In this case, agent $i$ picked both of their items before agent $j$ picked their second item, and so $v_i(\{g_i\}) \ge v_i(\{y_j\})$ and $v_i(\{h_i\}) \ge v_i(\{y_j\})$. Again, by \Cref{lem:additive-pair-demand-envy}, we have $v_i(\{g_i,h_i\}) \geq \mu_i(\{g_i,h_i, x_j, y_j\}) \geq \mu_i(X_i \cup X_j)$. 
   
   This completes the proof of the theorem.
\end{proof}

Since \Cref{alg:pair-demand-pmms} runs in polynomial time (in $n$ and $m$), we obtain the following corollary.

\begin{corollary}\label{cor:compute_pair_demand}
For any instance with pair-demand valuations, a PMMS allocation can be computed in polynomial time.
\end{corollary}

\section{Conclusion and Open Problems}
\label{sec:conclusion}

Our work has provided several new insights into the existence of EFX and PMMS allocations. In particular, we established the first existence results for the less-explored notion of PMMS.
We conclude by highlighting several interesting open problems that arise from our findings.

We begin with a natural strengthening of our nonexistence result in \Cref{thm:separation}. There, we show that a single MMS-feasible valuation is insufficient to guarantee the existence of a PMMS allocation in the three-agent setting. A compelling question is whether two MMS-feasible valuations suffice.

\vspace{\problemspace}
\noindent \textbf{Open Problem 1}: \emph{Does PMMS exist for three agents, where two have MMS-feasible valuations and the third has a monotone valuation?}
\vspace{\problemspace}

Next, in \Cref{prop:pmmsnonexistance}, we  presented a PMMS counterexample with $n + \mathcal{O}(\log{n})$ items. It would be interesting to determine whether this number can be reduced to a constant.

\vspace{\problemspace}
\noindent \textbf{Open Problem 2}: \emph{Does there exist a constant $C > 0$ such that for every $n \geq 2$, there is an instance with $n$ agents and $n + C$ items that admits no PMMS allocation?}
\vspace{\problemspace}

In \Cref{sec:bivalued}, we showed that EFX allocations always exist for personalized bivalued valuations, and that PMMS allocations exist when these valuations are additionally factored. An important direction for future work is to determine whether PMMS allocations exist more generally, without the factored assumption.

\vspace{\problemspace}
\noindent \textbf{Open Problem 3}: \emph{Does PMMS exist for (not necessarily factored) bivalued valuations?}
\vspace{\problemspace}

In \Cref{sec:binary}, we showed that PMMS allocations always exist for binary-valued valuations. It would be interesting to investigate whether PMMS allocations also exist for (MMS-feasible) binary-marginal valuations.

\vspace{\problemspace}
\noindent \textbf{Open Problem 4}: \emph{Does PMMS exist for all MMS-feasible binary-marginal valuations?}
\vspace{\problemspace}

 Notably, \citet*{BarmanV21} show that for matroid-rank valuations, a subclass of binary-marginal valuations, there always exists a partial PMMS allocation (i.e., one where some items may remain unallocated) that also maximizes utilitarian social welfare. We further observe that for binary additive valuations \cite[e.g.,][]{DBLP:conf/wine/0002PP020}, where $v_i(\{g\}) \in \{0,1\}$ for $g \in M$ and $v_i(S) = \sum_{g \in S} v_i(\{g\})$ for $S \subseteq M$, the existence of a PMMS allocation is automatically implied by the existence of an EFX allocation; see \Cref{sec:binary_additive} for further details.

We also believe it is worth exploring whether our Cut-and-Choose-Graph procedure can be applied to other valuation classes beyond binary-valued settings.

In \Cref{sec:pairdemand}, we proved that PMMS allocations exist for pair-demand valuations where an agent's valuation is additive over their two most valued items. A natural generalization is the non-additive setting, where $v_i(S) = \max_{T \subseteq S : |T| \leq 2} v_i(T)$. It remains open whether PMMS allocations exist in this broader class, under the additional assumption of MMS-feasibility.

\vspace{\problemspace}
\noindent \textbf{Open Problem 5}: \emph{Does PMMS  exist for any MMS-feasible pair-demand valuations?}
\vspace{\problemspace}

Moreover, we did not establish the existence of EFX allocations for (additive) pair-demand valuations, leaving another important open question.

\vspace{\problemspace}
\noindent \textbf{Open Problem 6}: \emph{Does EFX exist for additive pair-demand valuations?}
\vspace{\problemspace}

\bibliographystyle{abbrvnat}
\bibliography{praca}

\appendix

\section{Missing Proofs of Section 3}\label{sec:proofsnonexistence}

\thmseparation*
\begin{proof}
We construct an instance with three agents and six items. Agent 3 has an additive valuation function given by $v_3(j) = 100 + j$ for any item $j \in \{1, \ldots, 6\}$. 

For each $i \in {1, 2}$, we define a monotone valuation function $v_i$ as follows. Set $v_i(S) = 1$ for all sets $S$ with $|S| = 1$, and $v_i(S) = 7$ for all sets with $|S| \ge 3$. For sets $S$ with $|S| = 2$, the values $v_i(S)$ are specified in \Cref{tab:valuations-3-agents-6-items-compact}. Note that for every $S$ with $|S| = 2$, we have $v_i(S) \in \{2, 3, 4, 5, 6\}$.

Fix any allocation $X = \langle X_1, X_2, X_3 \rangle$. We show that it violates the PMMS condition by considering the following cases.

\textit{Case 1: Some agent $i$ receives an empty bundle.} In this case, there exists some $j \ne i$ such that $|X_j| \ge 2$. Since all valuations assign positive value to nonempty bundles, it follows that
$\mu_i(X_i \cup X_j, 2) = \mu_i(X_j, 2) > v_i(X_i) = 0$,
which violates the PMMS condition.

\textit{Case 2: Some agent $i$ receives a singleton bundle.}
In this case, there exists an agent $j \ne i$ such that $|X_j| \ge 3$. If $i = 3$, then $v_3(X_3) < 200$, while $\mu_3(X_i \cup X_j, 2) > 200$, violating the PMMS condition. If $i \in \{1, 2\}$, then $v_i(X_i) = 1$, while $\mu_i(X_i \cup X_j, 2) \geq 2$, since all 2-element bundles have value at least $2$ and $|X_i \cup X_j| \ge 4$.

\textit{Case 3: Every agent receives a pair of items.}  
There are $6! / 2^3 = 90$ possible allocations in this case. We have verified computationally that none of these allocations satisfies the PMMS condition. This can be visualized as follows.

Consider the following graph construction. For every agent $i$, we create a node for each $2$-element subset $S \subseteq M$, labeled as $S_i$. We add an edge between nodes $S_i$ and $T_j$ (for $i \neq j$ and $S \cap T = \emptyset$) if and only if both $v_i(S) \ge \mu_i(S \cup T, 2)$ and $v_j(T) \ge \mu_j(S \cup T, 2)$ hold.
In other words, an edge exists between $S_i$ and $T_j$ if agents $i$ and $j$ would not PMMS-envy each other when allocated bundles $S$ and $T$, respectively.

We present this graph for the given instance in \cref{fig:envy-graph-counterexample}. Nodes with no adjacent edges are omitted from the visualization. Clearly, a PMMS allocation exists if and only if this graph contains a triangle. In the given graph, no such triangle exists.
\end{proof}

\begin{table}[t]
  \centering
  \setlength{\tabcolsep}{4pt}       %
  \renewcommand{\arraystretch}{1.1} %

  \begin{tabular}{|l|*{5}{c|}}
    \hline
      & $\{1,2\}$ & $\{1,3\}$ & $\{1,4\}$ & $\{1,5\}$ & $\{1,6\}$  \\
    \hline
    $v_1$   & 6         & 5         & 2         & 2         & 4                \\
    $v_2$   & 3         & 5         & 2         & 2         & 3                 \\
    \hline
  \end{tabular}

  \vspace{2ex}  %

  \begin{tabular}{|l|*{5}{c|}}
    \hline
      & $\{2,3\}$ & $\{2,4\}$ & $\{2,5\}$ & $\{2,6\}$ & $\{3,4\}$ \\
    \hline
    \centering $v_1$   & 2         & 3         & 5    & 4 & 4     \\
    \centering $v_2$   & 2         & 2         & 5    & 3 & 2    \\
    \hline
  \end{tabular}

   \vspace{2ex}  %

  \begin{tabular}{|l|*{5}{c|}}
    \hline
      & $\{3,5\}$ & $\{3,6\}$ & $\{4,5\}$ & $\{4,6\}$ & $\{5,6\}$  \\
    \hline
    \centering $v_1$       & 6         & 5         & 4         & 6         & 3                \\
    \centering $v_2$          & 4         & 2         & 3         & 5         & 4               \\
    \hline
  \end{tabular}

  \caption{The values assigned to pairs of items in the proof of \Cref{thm:separation} are organized such that the columns correspond to item pairs and the rows correspond to agents.}
  \label{tab:valuations-3-agents-6-items-compact}
\end{table}

\begin{figure}[htbp]
    \centering
    \includegraphics[width=\textwidth]{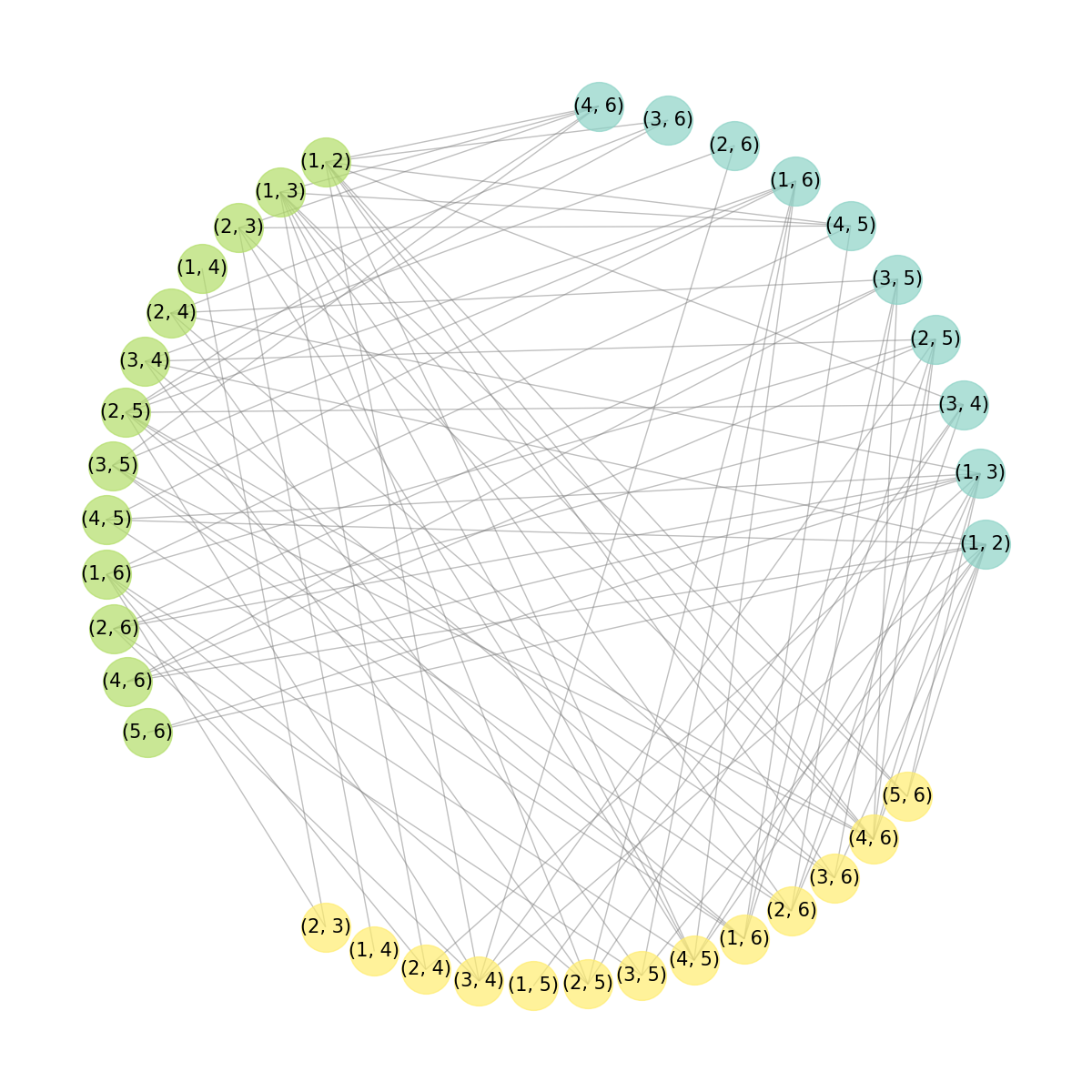}
    \caption{An illustration of the triangle-free graph considered in the proof of \Cref{thm:separation}. The blue, green, and yellow nodes correspond to 2-element subsets allocated to agents 1, 2, and 3, respectively. Note that we omit all nodes with no adjacent edges; for example, the blue node corresponding to the allocation $\{1,4\}$ for agent $1$ is omitted.}
    \label{fig:envy-graph-counterexample}
\end{figure}

\section{Missing Proofs of Section 4}\label{sec:proofsfrombivalued}

\lemprincipallemma*

    \begin{proof}
        Both parts of the lemma follow from the fact that high-value items are allocated via maximal matching. 
        Let $g$ be the item with $v_i(\{g\}) = a_i$  allocated to some agent in round $r_i$.
        
        For the first point, suppose for contradiction that agent $i$ was allocated an item worth $b_i$ in some round $r < r_i$. 
Since $r < r_i$, item $g$ was not yet allocated by round $r$. Because agent $i$ received a low-value item in round $r$, they were not matched in the matching selected during that round. However, since $g$ was still available, the edge $(i, g)$ could have been added to extend the matching in round $r$, contradicting the maximality of the matching chosen in \Cref{line:matching}.

       For the second part of the lemma, consider a round $r$ during which agent $i$ is frozen. By the condition in \Cref{line:if_unmatched}, agent $i$ can only be frozen if they are matched to some item and there exists another agent $j$ in the same connected component $C$ of $G$ who is unmatched. Since $i$ and $j$ belong to the same component, there exists an alternating path $i_0, g_1, i_1, g_2, i_2, \ldots, g_k, i_k$ in $G$, where $i_0 = j$, $i_k = i$,  and $(i_w, g_{w+1}) \in G$ for all $0 \leq w < k$, and $(i_w, g_w) \in G$ for all $0 < w \leq k$. 
       Along this path, agents $i_1, i_2, \ldots, i_k$ are matched to items $g_1, g_2, \ldots, g_k$, respectively. See \Cref{fig:alternating-path-contradiction} for an illustration of the alternating path used in the argument.

        Now, suppose $r < r_i$. Then the high-value item $g$ allocated in round $r_i$ has not yet been allocated in round $r$. In that case, we can increase the size of the matching by reassigning agent $i$ to item $g$, and reassigning agents $i_0, i_1, \ldots, i_{k-1}$ to items $g_1, g_2, \ldots, g_k$, respectively. This contradicts the maximality of the matching selected in \Cref{line:matching}. Therefore, agent $i$ cannot be frozen before round $r_i$.

        Furthermore, by \Cref{line:if_unmatched}, an agent can only get frozen in a round in which they are matched to a high-value item. Thus, agent $i$ must receive an item worth $a_i$ in that round. She also cannot get frozen in any round $r > r_i$, since by that point there are no remaining items of value $a_i$ available to her.

Moreover, agent $i$ is frozen for exactly $\lfloor t - 1 \rfloor$ rounds, where $t$ is the maximum value of $a_j / b_j$ among unmatched agents $j$ in the connected component to which $i$ belongs. By \Cref{lem:max-le-min}, the freezing time $t$ is at most $\lfloor a_i / b_i - 1\rfloor$. This completes the proof of the lemma.

For the third part, first consider the case where agent $j$ belongs to the same connected component as agent $i$ in the graph constructed during round $r_i$. By the construction of the algorithm, this implies that $F_i = F_j$.

Otherwise, if agents $i$ and $j$ belong to different connected components, then agent $j$ does not have an edge to item $X_{j,r_i}$ in the constructed graph. Hence, item $X_{j,r_i}$ must have been allocated to agent $j$ via the operation in \Cref{line:remitems} and was therefore unmatched during round $r_i$. We will now argue that agent $i$'s connected component contains no unmatched agents, which immediately implies that $F_i = F_j = \emptyset$.

Suppose, for the sake of contradiction, that there exists an unmatched agent $h$ in agent $i$'s connected component. As in the proof of the second part of the lemma, since $i$ and $h$ are in the same component, there exists an alternating path $i_0, g_1, i_1, g_2, i_2, \ldots, g_k, i_k$ in $G$, where $i_0 = h$, $i_k = i$, $(i_w, g_{w+1}) \in G$ for all $0 \leq w < k$, and $(i_w, g_w) \in G$ for all $0 < w \leq k$. Along this path, the agents $i_1, i_2, \ldots, i_k$ are matched to items $g_1, g_2, \ldots, g_k$, respectively.

Now, observe that we can increase the size of the matching by reassigning agent $i$ to item $X_{j,r_i}$, and reassigning agents $i_0, i_1, \ldots, i_{k-1}$ to items $g_1, g_2, \ldots, g_k$, respectively. This contradicts the maximality of the matching selected in \Cref{line:matching}, thus completing the proof of the lemma.
   \end{proof}

\section{Binary Additive Valuations}\label{sec:binary_additive}

In this section, we briefly discuss the case of binary additive valuations, where $v_i(\{g\}) \in \{0,1\}$ for all items $g \in M$ and $v_i(S) = \sum_{g \in S} v_i(\{g\})$ for all bundles $S \subseteq M$. In the following lemma, we argue that, in contrast to general non-degenerate instances, PMMS is strictly weaker than EFX under binary additive valuations.

\begin{lemma}
For any instance with binary additive valuations, any allocation that satisfies EFX also satisfies PMMS.
\end{lemma}

\begin{proof}
Let $X = \langle X_1, \ldots, X_n \rangle$ be an EFX allocation. Fix any pair of agents $i$ and $j$. By EFX, we have $v_i(X_i) \geq v_i(X_j \setminus \{g\})$ for all $g \in X_j$.

First, if $v_i(X_i) \geq v_i(X_j)$, then clearly $v_i(X_i) \geq (v_i(X_i) + v_i(X_j))/2 = v_i(X_i \cup X_j)/2 \geq \mu_i(X_i \cup X_j)$, so the PMMS condition holds.

Otherwise, we must have $v_i(X_j) = v_i(X_i) + 1$ and $v_i(\{g\}) = 1$ for all $g \in X_j$. Moreover, $v_i(X_i \cup X_j) = 2 \cdot v_i(X_i) + 1$, so the fair share is $\mu_i(X_i \cup X_j) = \left\lfloor v_i(X_i \cup X_j)/ {2} \right\rfloor = v_i(X_i)$, and again PMMS is satisfied.
\end{proof}

For binary additive valuations, the existence of EFX allocations has been established by \citet*{DBLP:conf/wine/0002PP020,DBLP:conf/aaai/BabaioffEF21,DBLP:journals/teco/BenabbouCIZ21}. Combined with the lemma above, this implies the existence of PMMS allocations in such settings.

\begin{corollary}
For any instance with binary additive valuations, a PMMS allocation exists.
\end{corollary}

\end{document}

%% file: figures/simplethreeagents.tex
\begin{figure}[t]
    \centering
    \begin{minipage}{0.49\textwidth}
        \centering  
         \usetikzlibrary {shapes.geometric}
  \begin{tikzpicture}
        \node[circle, draw] (n1) at (0,0) {1};
        \node[circle, draw] (n2) at (2.5,0) {1};
        \node[circle, draw] (n3) at (2.5,2.5) {1};
        \node[circle, draw] (n4) at (0,2.5) {1};

      \draw[blue] (n1) to node[below] {3} (n2);
      \draw[blue] (n3) to node[above] {3} (n4);
      \draw[red] (n1) to node[left] {3} (n4);
      \draw[red] (n2) to node[right] {3} (n3);
    \end{tikzpicture}

            \smallbreak
    (i) An instance with two agents.
        \end{minipage}
            \hfill
    \hfill
    \begin{minipage}{0.49\textwidth}
        \centering
    \usetikzlibrary {shapes.geometric}
  \begin{tikzpicture}
        \node[circle, draw] (n1) at (0,0) {1};
        \node[circle, draw] (n2) at (2.5,0) {1};
        \node[circle, draw] (n3) at (2.5,2.5) {1};
        \node[circle, draw] (n4) at (0,2.5) {1};
        \node [star, minimum size=1mm, draw, inner sep=1mm, fill=yellow]
          (n5) at (4,1.25) {2};

      \draw[blue] (n1) to node[below] {3} (n2);
      \draw[blue] (n3) to node[above] {3} (n4);
      \draw[green!80!black] (n1) to node[right=0.3] {3} (n3);
      \draw[green!80!black] (n2) to node[left=0.3] {3} (n4);
      \draw[red] (n1) to node[left] {3} (n4);
      \draw[red] (n2) to node[right] {3} (n3);
    \end{tikzpicture}
  \smallbreak
  (ii) An instance with three agents.
  \end{minipage}

  \caption{Illustrations of the instances used in the proof of \Cref{prop:pmmsnonexistance}.
In both constructions, we set $k = 2$.
Nodes in the graph represent items, where circles denote common items.
The numbers inside the nodes indicate how much each agent values the corresponding item.
Red edges connect pairs of items that the first agent values at 3; similarly, blue and green edges indicate pairs valued at 3 by the second and third agents, respectively. }
      \label{fig:pmms-3-agents-counterexample}
\end{figure}

%% file: algos/bivaluedvaluations.tex
\begin{figure}[p]               
\centering

\begin{minipage}{\textwidth}

\begin{algorithm}[H]
     \SetKwInOut{Input}{Input}
  \SetKwInOut{Output}{Output}
    \LinesNumbered
    \DontPrintSemicolon
    \Input{A personalized bivalued instance $I = \langle N, M, \mathcal{V} \rangle$}      
    \Output{An EFX allocation $X = \langle X_1, \ldots, X_n \rangle$}

    Set $P \gets M$ \tcp*{$P$ is the set of unallocated items}

    {Set $w_i \gets 0$ for $i\in \{1, \ldots, n\}$ \tcp*{$w_i$ is agent {$i$'s} priority}}

    Set $L_r \gets N$ for $r \in \{1, \ldots, m\}$ \tcp*{$L_r$ is the set of agents active in round $r$}

    Set $X_i \gets \emptyset$ for $i \in \{1, \ldots, n\}$ \tcp*{$X_i$ is the allocation of agent $i$}

    \While{$P \neq \emptyset$}{
        Set $r$ to be the current round (iteration) number

        Construct a weighted bipartite graph $G$ between $L_r$ and $P$

        \For{each $i \in L_r$ and $g \in P$ such that $v_i(\{g\}) = a_i$}{
            Add an edge from $i$ to $g$ with weight $a_i / b_i$
        }

        Find a maximal (in terms of size) matching in $G$ that maximizes the total weight\label{line:matching}

        \For{each matched pair $(i, g)$}{
            Add $g$ to $X_i$

            Remove $g$ from $P$
        }

        \For{each connected component $C$ of $G$}{
            Let $U$ be the set of unmatched agents in $C$

            \If{$U \neq \emptyset$}{\label{line:if_unmatched}
                Let $t$ be the maximum value of $a_i / b_i$ for $i \in U$\label{lin:freezingtime}

                Remove all matched agents of $C$ from $L_{r+j}$ for $j \in \{1, \ldots, \lfloor t - 1 \rfloor\}$ \tcp*{Freeze matched agents}

                {$w_i\gets r$ for each matched $i$ in $C$} \label{lin:settingpriority}
            }
        }

        \For{each unmatched agent $i$ {in increasing order w.r.t. $w_i$}}{ \label{lin:unmatched_order}
            Add any remaining good $g$ to $X_i$\label{line:remitems}

            Remove $g$ from $P$
        }
    }

    \Return{$X = \langle X_1, \ldots, X_n \rangle$}

    \caption{Personalized Match-and-Freeze Algorithm}
    \label{alg:personalized-bivalued-pmms}
\end{algorithm}

\end{minipage}

\vspace{1em}

\begin{minipage}{\textwidth}
\renewcommand{\arraystretch}{1.4} 
\newcommand{\circled}[1]{
  \tikz[baseline=-0.8ex]            
    \node[draw,circle,inner sep=1pt]{#1};}

\newcommand{\crosscell}{%
  \tikz[baseline=(X.base)] \node[draw=black, line width=0.4mm, cross out, minimum width=1em, minimum height=1em, inner sep=0pt] (X) {};
}
\begin{center}
    \begin{tabular}{|c|c|c|c|c|c|c|}
  \hline
  \textbf{Agents} & \textbf{Round 1} & \textbf{Round 2} & 
  \textbf{Round 3} & \textbf{Round 4} & \textbf{Round 5} & \textbf{Round 6} \\ \hline
  Agent 1 ({$a_1 = 2.5$}) & 1 & 1 & 1 & 1 & 1 & 1 \\ \hline
  Agent 2 ($a_2 = 3$) & \circled{3} & ---
  & 1 & 1 & 1 &   \\ \hline
  Agent 3 ($a_3 = 4$) & 1 & 1 & 1 & 1 & 1 & 1 \\ \hline
  Agent 4 ($a_4 = 5$) & \circled{5} & --- & --- & --- & 1 & \\ \hline
\end{tabular}
  \end{center}
   \captionof{table}{An example execution of \Cref{alg:personalized-bivalued-pmms} is presented for an instance with 4 agents and 18 items. The agents have bivalued valuations with $a_1 = 2.5$, $a_2 = 3$, $a_3 = 4$, $a_4 = 5$, and $b_i = 1$ for $i \in \{1, 2, 3, 4\}$. The items are ${x, y, z_1, \ldots, z_{16}}$. Agents 1 and 2 assign the higher value to item $x$ and no other item, while agents 3 and 4 assign the higher value to item $y$ and no other item.
Each cell represents the value of the item that a given agent received in a particular round. 
Cells with a circled value indicate rounds in which an agent received a high-value item. Cells corresponding to rounds in which an agent is frozen contain a dash, and cells for rounds after the last item has been allocated are left empty.}   
  \label{tab:allocation-rounds}
\end{minipage}

\end{figure}

%% file: figures/alternating_path.tex
\usetikzlibrary{positioning}  %

\begin{figure}
  \centering
  \begin{tikzpicture}[
      scale=1.2,
      vertex/.style={circle, draw, minimum size=0.7cm, inner sep=1pt},
      item/.style={rectangle, draw, minimum size=0.6cm, inner sep=2pt},
      unmatched/.style={vertex, fill=white},
      matched/.style={vertex, fill=gray!20},
      matched edge/.style={line width=1.5pt, color=red},
      alternating edge/.style={line width=1.5pt, dashed},
      node distance=2cm,
      >=latex
    ]
    \node[unmatched] (i0)    {$i_0$};
    \node[matched]   (i1)    [right=of i0]    {$i_1$};
    \node[matched]   (i2)    [right=of i1]    {$i_2$};
    \node[matched]   (i3)    [right=of i2]    {$i_3$};
    \node[matched]   (idots) [right=of i3]    {$\dots$};
    \node[matched]   (ik)    [right=of idots] {$i_k$};

    \node[item] (g1)    [below=of i0]    {$g_1$};
    \node[item] (g2)    [below=of i1]    {$g_2$};
    \node[item] (g3)    [below=of i2]    {$g_3$};
    \node[item] (gdots) [below=of i3]    {$\dots$};
    \node[item] (gk)    [below=of idots] {$g_k$};

    \draw[alternating edge] (i0)    -- node[midway, left]  {$\frac{a_j}{b_j}$}         (g1);
    \draw[alternating edge] (i1)    -- node[midway, left] {$\frac{a_{i_1}}{b_{i_1}}$}    (g2);
    \draw[alternating edge] (i2)    -- node[midway, left] {$\frac{a_{i_2}}{b_{i_2}}$}    (g3);
    \draw[alternating edge] (i3)    -- node[midway, left] {$\frac{a_{i_3}}{b_{i_3}}$}    (gdots);
    \draw[alternating edge] (idots) -- node[midway, left] {$\frac{a_{i_{k-1}}}{b_{i_{k-1}}}$} (gk);

    \draw[matched edge] (i1) -- node[midway, left, xshift=-0.1cm] {$\frac{a_{i_1}}{b_{i_1}}$}  (g1);
    \draw[matched edge] (i2) -- node[midway, left, xshift=-0.1cm] {$\frac{a_{i_2}}{b_{i_2}}$}  (g2);
    \draw[matched edge] (i3) -- node[midway, left, xshift=-0.1cm] {$\frac{a_{i_3}}{b_{i_3}}$} (g3);
    \draw[matched edge] (idots) --  (gdots);
    \draw[matched edge] (ik) -- node[midway, left, xshift=-0.1cm] {$\frac{a_i}{b_i}$}       (gk);

  \end{tikzpicture}
  \caption{
    An illustration of the alternating paths used in the proofs of \Cref{lem:principal-lemma,lem:max-le-min}.
Solid red edges represent edges included in the matching, while dashed black edges represent edges not included in the matching. Circles represent nodes corresponding to agents, while squares represent nodes corresponding to items. Grey agent nodes indicate matched agents, while white agent nodes indicate unmatched agents.
Edge labels indicate their corresponding weights.
  }
  \label{fig:alternating-path-contradiction}
\end{figure}
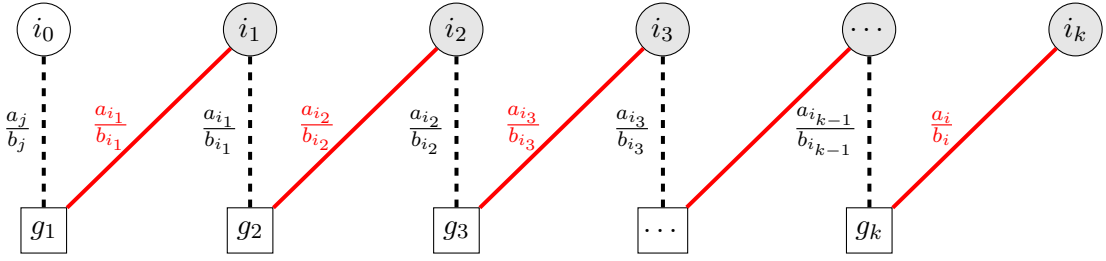

%% file: algos/binaryvaluations.tex
\begin{figure}[!htbp]
\centering

  \begin{minipage}{\textwidth}
    \centering
    \begin{algorithm}[H]
  \SetKwInOut{Input}{Input}
  \SetKwInOut{Output}{Output}
  \LinesNumbered
  \DontPrintSemicolon

  \Input{A binary-valued MMS-feasible instance $I = \langle N, M, \mathcal{V} \rangle$}
  \Output{A PMMS allocation $X = \langle X_1, \ldots, X_n \rangle$}

  Let $X = \{X_1, \ldots, X_n\}$ be an arbitrary allocation\;

  \While{$X$ is not PMMS}{\label{line:binary_while}
    Choose $s \in N$ such that $v_s(X_s) < \pmms_s(X_s \cup X_j)$ for some $j \in N$\;\label{line:binarychoiceofs}

    Construct the cut-and-choose graph $\pi$ for agent $s$ and allocation $X$\;

    Set $i_0 \gets s$\;
    Set $k \gets 0$\;

    \While{$\pi(i_k) \notin \{i_0, \ldots, i_k\}$}{
      Set $i_{k+1} \gets \pi(i_k)$\;
      
      Set $k \gets k + 1$\;
    }

    \eIf{$\pi(i_k) = i_0$}{\label{line:binary_if}
      Set $X_i \gets X_{\pi(i)}$ for all $i \in \{i_0, \ldots, i_k\}$
      \tcp*{Case 1: Cycle}
      
    }{
      Let $w$ be such that $\pi(i_k) = i_w$
 \tcp*{Case 2: Lollipop}
     
      Let $(A, B)$ be a maximin partition of $X_{i_0} \cup X_{i_w}$ by $v_{i_k}$\;\label{line:binarychoiceofaandb}

        \If{$v_{i_{w-1}}(A) < v_{i_{w-1}}(B)$}{
      Swap $A$ and $B$\;\label{line:binaryswappingofaandb}
    }
    
      Set $X_i \gets X_{\pi(i)}$ for all $i \in \{i_0, \ldots, i_{w-2}\} \cup \{i_w, \ldots, i_{k-1}\}$\;

      Set $X_{i_{w-1}} \gets A$\;
      
      Set $X_{i_k} \gets B$\;
    }
  }

  \Return{$X = \langle X_1, \ldots, X_n \rangle$}

  \caption{Cut-and-Choose-Graph Procedure}
  \label{alg:binary-pmms}
\end{algorithm}
\end{minipage}

  \vspace{1em}

\begin{minipage}{\textwidth}
    \centering

    \begin{center}
    \centering
    \begin{minipage}{0.41\textwidth}
        \centering
        \begin{tikzpicture}[
            every node/.style={circle, draw, minimum size=1cm, inner sep=0pt, align=center},
            >=stealth, thick
        ]
          \node (j) at (0,1.5) {$s$};
          \node (i1) at (1.75,1.5) {$i_1$};
          \node (i2) at (3.5,1.5) {$i_2$};

          \node (i5) at (0,0) {$i_k$};
          \node (i4) at (1.75,0) {$\ldots$};
          \node (i3) at (3.5,0) {$i_3$};

          \draw[->] (j) -- (i1);
          \draw[->] (i1) -- (i2);
          \draw[->] (i2) -- (i3);
          \draw[->] (i3) -- (i4);
          \draw[->] (i4) -- (i5);
          \draw[->] (i5) -- (j);
        \end{tikzpicture}
        \smallbreak
        (i) Case 1: Cycle.
    \end{minipage}
    \hfill
    \begin{minipage}{0.58\textwidth}
        \centering
        \begin{tikzpicture}[
            every node/.style={circle, draw, minimum size=1cm, inner sep=0pt, align=center},
            >=stealth, thick
        ]
          \node (j) at (-1.75,1.5) {$s$};
          \node (i05) at (0,1.5) {$\ldots$};
          \node (i1) at (1.75,1.5) {$i_{w-1}$};
          \node (i2) at (3.5,1.5) {$i_w$};
          \node (i3) at (5.25,1.5) {$\ldots$};

          \node (i5) at (3.5,0) {$i_k$};
          \node (i4) at (5.25,0) {$i_{k-1}$};

          \draw[->] (j) -- (i05);
          \draw[->] (i05) -- (i1);
          \draw[->] (i1) -- (i2);
          \draw[->] (i2) -- (i3);
          \draw[->] (i3) -- (i4);
          \draw[->] (i4) -- (i5);
          \draw[->] (i5) -- (i2);
        \end{tikzpicture}
        \smallbreak
        (ii) Case 2: Lollipop.
    \end{minipage}
\end{center}
\end{minipage}

\caption{An illustration of the operations performed by the Cut-and-Choose-Graph Procedure (\Cref{alg:binary-pmms}). The nodes represent agents, and the arrows depict the cut-and-choose graph. Subfigure (i) illustrates Case (1), where the if-condition in \Cref{line:binary_if} holds, and the new allocation is obtained by assigning each agent $i$ in the cycle the bundle of the agent to whom $i$ points. Subfigure (ii) illustrates Case (2), where the if-condition in \Cref{line:binary_if} does not hold. In this case, the new allocation is obtained by giving all agents in the cycle—except for agents $i_{w-1}$ and $i_k$—the bundles they are pointing to, and having agents $i_{w-1}$ and $i_k$ divide the two bundles $X_s$ and $X_{i_w}$ between themselves using the cut-and-choose method.}
\label{fig:binaryvaluations}
\end{figure}

%% file: algos/pmmspairdemand.tex
\begin{algorithm}[t]
     \SetKwInOut{Input}{Input}
     \SetKwInOut{Output}{Output}
     \LinesNumbered
     \DontPrintSemicolon
     \Input{A pair-demand instance $I = \langle N, M, \mathcal{V} \rangle$ with $|M| \geq 2|N|$}
     \Output{A PMMS allocation $X = \langle X_1, \ldots, X_n \rangle$}
     Set $X_i \gets \emptyset$ for $i \in \{1, \ldots, n\}$ \tcp*{$X_i$ is the allocation of agent $i$}
     Set $P \gets M$ \tcp*{$P$ is the set of unallocated items}
     \For{$i \in N$ in increasing order of indices}{
         Let $g_i \in \argmax_{g \in P} v_i(\{g\})$ \;
         Add $g_i$ to $X_i$ \;
         Remove $g_i$ from $P$\;
     }
     \For{$i \in N$ in decreasing order of indices}{
         Let $h_i \in \argmax_{h \in P} v_i(\{h\})$ \;
         Add $h_i$ to $X_i$ \;
         Remove $h_i$ from $P$\;
     }
     Allocate all items left in $P$ to any player\;
     \Return{$X = \langle X_1, \ldots, X_n \rangle$}
     \caption{Reversed Round-Robin Algorithm}
     \label{alg:pair-demand-pmms}
\end{algorithm}